\documentclass[conference]{IEEEtran}

\usepackage{amsmath, amsthm, amssymb}
\usepackage{latexsym}
\usepackage{graphicx}
\usepackage{verbatim}
\usepackage{mathrsfs}
\usepackage{float}
\usepackage[lofdepth, lotdepth]{subfig}
\usepackage{array}
\usepackage{url}
\usepackage{algorithm}
\usepackage[noend]{algorithmic}
\usepackage{cite}
\usepackage[table]{xcolor}
\usepackage{enumerate}
\usepackage{eucal}
\usepackage{stmaryrd}
\usepackage{mathtools}
\usepackage{pgf}
\usepackage{tikz}
\usetikzlibrary{arrows,automata}
\usepackage[latin1]{inputenc}
\usetikzlibrary{automata,positioning}

\usepackage[top=0.75in, bottom=0.8in, left=0.5in, right=0.5in]{geometry}
\usepackage{mathtools}
\DeclarePairedDelimiter\ceil{\lceil}{\rceil}
\DeclarePairedDelimiter\floor{\lfloor}{\rfloor}

\theoremstyle{plain}
\newtheorem{theorem}{Theorem}
\newtheorem{proposition}{Proposition}
\newtheorem{lemma}{Lemma}

\newtheorem{corollary}{Corollary}
\newtheorem{construction}{Construction}
\theoremstyle{definition}
\newtheorem{definition}{Definition}
\newtheorem{example}{Example}

\newcommand{\C}{{\mathcal C}}

\newcommand{\cE}{{\mathcal E}}

\newcommand{\cL}{{\mathcal L}}

\newcommand{\cS}{{\mathcal S}}

\newcommand{\bL}{{\boldsymbol L}}

\newcommand{\bp}{{\boldsymbol p}}
\newcommand{\bu}{{\boldsymbol u}}
\newcommand{\bv}{{\boldsymbol v}}

\newcommand{\by}{{\boldsymbol y}}

\newcommand{\bw}{{\boldsymbol{w}}}

\newcommand{\bx}{{\boldsymbol{x}}}

\newcommand{\vu}{{\sf u}}

\renewcommand{\ge}{\geqslant}
\renewcommand{\le}{\leqslant}
\newcommand{\lek}{{{\le}k}}

\newcommand{\rate}{{\rm rate}}
\newcommand{\Irr}{{\rm Irr}}

\newcommand{\et}{{\emph{et al.}}}

\begin{document}

\pagestyle{empty}

\title{Efficient Encoding/Decoding of Irreducible Words for Codes Correcting Tandem Duplications}

\author{
   \IEEEauthorblockN{
	Yeow Meng Chee,
	Johan Chrisnata,
	Han Mao Kiah, and
	Tuan Thanh Nguyen}
   \IEEEauthorblockA{
   School of Physical and Mathematical Sciences,
	Nanyang Technological University, Singapore\\
email: \{{ymchee}, {jchrisnata}, {hmkiah}, {nguyentu001}\}@ntu.edu.sg\\[-5mm]
 } 
}
\maketitle

\hspace{-3mm}\begin{abstract}
Tandem duplication is the process of inserting a copy of a segment of DNA adjacent to the original position. Motivated by applications that store data in living organisms, Jain \et\ (2017) proposed the study of codes that correct tandem duplications. 
Known code constructions are based on {\em irreducible words}.

We study efficient encoding/decoding methods for irreducible words.
First, we describe an $(\ell,m)$-finite state encoder
and show that when  $m=\Theta(1/\epsilon)$ and $\ell=\Theta(1/\epsilon)$, 
the encoder achieves rate that is $\epsilon$ away from the optimal.
Next, we provide ranking/unranking algorithms for irreducible words and 
modify the algorithms to reduce the space requirements for the finite state encoder.
\end{abstract}


\section{Introduction}

Advances in synthesis and sequencing technologies 
have made DNA macromolecules an attractive medium for digital information storage.
Besides being biochemically robust, DNA strands offer ultrahigh storage densities of $10^{15}$-$10^{20}$ bytes per gram of DNA, as demonstrated in recent experiments (see \cite[Table 1]{Yazdi.2017}).

These synthetic DNA strands may be stored {\em ex vivo} or {\em in vivo}. 
When the DNA strands are stored {\em ex vivo} or in a non-biological environment, 
code design takes into account the synthesising and sequencing platforms being used 
(see \cite{Yazdi.etal:2015b} for a survey of the various coding problems).
In contrast, when the DNA strands are stored {\em in vivo} or recombined with the DNA of a living organism,
we design codes to correct errors due to the biological mutations.

This work looks at the latter case, and 
specifically, examines codes that correct errors due to {\em tandem duplications}. 
Tandem duplications or repeats is one of the two common repeats found in the human genome \cite{Lander} 
and they are caused by slipped-strand mispairings \cite{Mundy}.
They occur in DNA when a pattern of one or more nucleotides is repeated and the repetitions are directly adjacent to each other. For example, consider the string or word {\tt AGTAGTCTGC}. The substring {\tt AGTAGT} is a tandem repeat, and we say that {\tt AGTAGTCTGC} is generated from {\tt AGTCTGC} by a {\em tandem duplication} of length three. 

Jain \et\ \cite{Jain} first proposed the study of codes that correct errors due to tandem duplications. 
In the same paper, Jain \et\ used {\em irreducible} words (see Section \ref{sec:prelim} for definition) to construct
 a family of codes that correct tandem duplications of lengths at most $k$, where $k\in\{2,3\}$.
While these codes are optimal in size for the case $k=2$, these codes are not optimal for $k=3$,
and in fact, Chee \et\ \cite{Chee:2017} constructed a family of codes with strictly larger size.
Recently, Jain \et\ \cite{Jain:2017} looked at other error mechanisms, and studied the capacity of these tandem-duplication systems in the presence of point-mutation noise (substitution errors).

In this paper, we look at encoding/decoding methods for irreducible words.
In particular, we provide polynomial-time algorithms that encodes either exactly the rates of irreducible words 
or close to the asymptotic rates of irreducible words.
While the encoding/decoding algorithms are standard in constrained coding \cite{ROTH}
and combinatorics literature \cite{Nijenhuis.2014}, 
our contribution is a detailed analysis of the space and time complexities of the respective algorithms.
Before we state the main results of the paper, we go through certain notations.

\subsection{Notation and Terminology}\label{sec:prelim}

Let $[n]$ denotes the set $\{1,2,\ldots, n\}$.
Let $\Sigma_q=\{0,1,\cdots q-1\}$ be an alphabet of $q\ge 2$ symbols.
For a positive integer $n$, let $\Sigma_q^n$ denote the set of all words of length $n$ over $\Sigma_q$,
and let $\Sigma_q^*$ denote the set of all words over $\Sigma_q$ with finite length.
Given two words ${\bx}, {\by} \in \Sigma_q^*$, we denote their concatenation by ${\bx\by}$. 

We state the {\em tandem duplication} rules. 
For integers $k\le n$ and $i\le n-k$, we define $T_{i,k} : \Sigma_q^{n} \to \Sigma_q^{n+k}$ such that
$T_{i,k}({\bx}) = \bu\bv\bv\bw$, where $\bx=\bu\bv\bw,\, |\bu|=i,\, |\bv|=k$.

If a finite sequence of tandem duplications of length at most $k$ is performed to obtain ${\by}$ from ${\bx}$, 
then we say that ${\by}$ is a ${\le}k$-\textit{descendant} of ${\bx}$, or 
$\bx$ is a $\lek$-\textit{ancestor} of ${\by}$ .
Given a word $\bx$, we define
the ${\le}k$-{\em descendant cone} of $\bx$ is the set of all $\lek$-descendants of $\bx$
and denote this cone by $D_{\lek}^*(\bx)$.

\begin{example}
Consider ${\bx}=01210$ over $\Sigma_3$. 
We have $T_{1,3}({\bx}) = 01211210$ and $T_{0,2}(01211210)=0101211210$.
So, $0101211210 \in D_{\le 3}^*({\bx})$. 
\end{example}
\begin{definition}[${\le}k$-Tandem-Duplication Codes] 
A subset ${\cal C} \subseteq \Sigma_q^n$ is a ${\le}k$-{\em tandem-duplication code} if 
for all ${\bx},{\by} \in {\cal C}$ and ${\bx} \ne {\by}$, we have that $D_{\le k}^*({\bx})\cap D_{\le k}^*({\by}) = \varnothing$.
We say that ${\cal C}$ is an $(n,{\le}k;q)$-TD code. 
\end{definition}

The {\em size} of $\C$ refers to $|\C|$, while the {\em rate} of $\C$ is given by $(1/n)\log_q|\C|$.
Given an infinite family $\{\C_n \mid \C_n \mbox{ is of length } n\}_{n=1}^\infty$, its {\em asymptotic rate}
is given by $\lim_{n\to\infty} (1/n)\log_q|\C_n|$.

\subsection{Irreducible Words}

Of interest is a family of tandem-duplication codes constructed by Jain \et \cite{Jain}.
Crucial to the code construction is the concept of irreducible words and roots. 

\begin{definition} A word is {\em $\lek$-irreducible} if it cannot be deduplicated into shorter words 
with deduplications of length at most $k$. We use ${\rm Irr}_{\lek}(n,q)$ to denote the set of all $\lek$-irreducible words  of length $n$ over $\Sigma_q$.
The $\lek$-ancestors of ${\bx} \in \Sigma_q^*$ that are $\lek$-irreducible words are called the $\lek$-\textit{roots} of 
${\bx}$. 
\end{definition}


\begin{construction}[Jain \et \cite{Jain}]\label{code:jain}
For $k\in\{2,3\}$ and $n \ge k$. An $(n,{\le}k; q)$-TD-code ${\cal C}(n, {\le}k;q)$ is given by
\begin{equation*}
{\C}(n,\lek;q)\triangleq\bigcup_{i=1}^{n} \left\{\xi_{n-i}({\bx})\mid {\bx} \in {\rm Irr}_{\le k}(i,q)\right\}.
\end{equation*} 
Here, $\xi_{i}({\bx})=\bx z^i$, where $z$ is the last symbol of $\bx$. 
\end{construction}

We point out certain advantages of Construction \ref{code:jain}.
\begin{enumerate}[(a)]
\item {\em Almost optimal rates}. Jain \et\  demonstrated that Construction \ref{code:jain} is optimal for $k=2$.
However, when $k=3$, Chee \et\ \cite{Chee:2017} provided constructions that achieve almost twice the size
in Construction \ref{code:jain} (see \cite[Table I]{Chee:2017}). 
Unfortunately, the asymptotic rate of the latter is the same as Construction~\ref{code:jain}.
Therefore, the set of irreducible words gives the best known asymptotic rates for $k=3$.

Furthermore, for $q\ge 5$ and $k=3$, the asymptotic rates of Construction \ref{code:jain} differs from a theoretical upper bound (see \cite[Proposition 4]{Chee:2017} and Table~\ref{rates}) by at most $0.01$. In other words,  Construction \ref{code:jain} is almost optimal in terms of rates.

\item {\em Linear-time decoding}. Consider $\bx\in{\C}(n,\lek;q)$ and we read $\by\in D_{\lek}^*(\bx)$. 
To retrieve the codeword $\bx$, we simply compute the $\lek$-root of $\by$ and extend the root if the root is shorter than $n$. Jain \et\ showed that there is at most one root when $k\in\{2,3\}$, while
Chee \et\ provided algorithms to compute these roots in linear time \cite{Chee:2017}.
\end{enumerate}

In view of these points, we study other practical aspects of Construction \ref{code:jain}.
Specifically, we look at {\em efficient encoding} of messages in $\Sigma_q^\ell$ to codewords in $\bx\in{\C}(n,\lek;q)$ 
for some $\ell<n$.

To this end, we look at the rates of $\C(n,\lek;q)$.
Let $I_{\lek}(n,q)\triangleq |{\rm Irr}_{\lek}(n,q)|$.
Then the size of ${\C}(n,\lek;q)$ is given by $\sum_{i=1}^n I_{\lek}(i,q)$.
Let $\rate_{\lek}(n,q)$ and $\rate_{\lek}(q)$ denote the  rate and asymptotic rate of ${\C}(n,\lek;q)$, respectively.
In other words, $\rate_{\lek}(n,q)\triangleq(1/n)\log_q|\C(n,\lek;q)|$
and $\rate_{\lek}(q)\triangleq \lim_{n\to\infty} \rate_{\lek}(n,q)$.
Jain \et\ observed that $\bigcup_{n=1}^{\infty} {\rm Irr}_{\lek}(n,q)$ is a regular language and 
hence, 
\begin{equation}\label{eq:rate}
\rate_{\lek}(q)= \lim_{n\to\infty} \frac{\log_q I_{\lek}(n,q)}{n}.
\end{equation}
Furthermore, using Perron-Frobenius theory (see \cite{ROTH}), Jain \et\ computed 
$\rate_{{\le}3}(3)$ to be approximately $0.347934$. 
In view of \eqref{eq:rate}, we look at encoding of the words in ${\rm Irr}_{\lek}(n)$ instead and
the extension of our encoding methods to ${\C}(n,\lek;q)$ is straightforward.

\subsection{Our Contributions}

We first develop a recursive formula for $I_{\lek}(n,q)$ and hence,
provide a formula for the asymptotic rate for $\C(n,\lek;q)$.
We then provide two efficient encoding methods
and use combinatorial insights provided by the recursive formula
to analyse the space and time complexities.

Specifically, our main contributions are as follows.
\begin{enumerate}[(A)]
\item We compute $\rate_{\lek}(q)$ for all $q$ and $k\in\{2,3\}$ in Section~\ref{sec:enumerate}.
\item In Section~\ref{sec:fse}, we propose an $(\ell,m)$-finite state encoder with rate $\ell/m$.
Furthermore, we show that we can choose the lengths $\ell$ and $m$ to be small and 
yet come close to the asymptotic rate. In particular, if we choose $m=\Theta(1/\epsilon)$ and $\ell=\Theta(1/\epsilon)$,
we showed that the rate is at least $\rate_{\lek}(q)$. 
Here, the running time for the encoder is linear in codeword length $n$ for constant $\epsilon$ .
\item Using bijections developed Section~\ref{sec:enumerate}, we provide a ranking/unranking algorithm 
that encodes with rate equal to $(1/n)\log_q({\rm Irr}_{\lek}(n,q))$ in Section~\ref{sec:rank}. 
This algorithm runs in $O(n^2)$ time using $O(n^2)$ space.
Furthermore, this ranking/unranking technique can be modified to reduce the space requirement to $O(m^2)$ 
in the $(\ell,m)$-finite state encoder.
\end{enumerate}

Due to space constraints, we present proofs and illustrate examples for the case $k=2$ and 
simply state the relevant results for $k=3$.
The detailed proofs are deferred to the full paper.

\section{Enumerating Irreducible Words}\label{sec:enumerate}

In this section, we compute $\rate_{\lek}(q)$ for all $q$ and $k\in\{2,3\}$ by 
obtaining a recursive formula for $I_\lek(n,q)$. 
While the Perron-Frobenius theory (see \cite{ROTH}) is sufficient to determine the asymptotic rates, 
the recursive formula is useful in the analysis of the finite state encoder in Section~\ref{sec:fse}
and the development of the ranking/unranking methods in Section~\ref{sec:rank}.

To this end, we partition the set of irreducible words into two classes and 
provide bijections from irreducible words of shorter lengths into them.
Specifically, notice that the suffix of an irreducible word is of the form 
either $aba$ or $abc$, where $a$, $b$, $c$ are distinct symbols.
Hence, we let $\Irr^{(s)}_\lek(2,n,q)$ and $\Irr^{(s)}_\lek(3,n,q)$ denote
the set of irreducible words with length-three suffixes that have two and three distinct symbols, respectively.

In the case $k=2$, we consider the following maps for $n\ge 4$,
\begin{align*}
\phi: \Irr_{\le{2}}(n-1,q)\times [q-2]&\to\Irr_{\le{2}}^{(s)}(3,n,q),\\
\psi: \Irr_{\le{2}}(n-2,q)\times [q-2]&\to\Irr_{\le{2}}^{(s)}(2,n,q).
\end{align*}

We first define $\phi$. If $\bx=x_1x_2\ldots x_{n-1}\in \Irr_{\le{2}}(n-1,q)$ and $i\in [q-2]$,
set $\sigma$ to be the $i$th element in $\Sigma_q\setminus \{x_{n-2},x_{n-1}\}$.
Then set $\phi(\bx,i)=x_1x_2\ldots x_{n-1}\sigma$.

For $\psi$, let $\bx=x_1x_2\ldots x_{n-2}\in \Irr_{\le{2}}(n-2,q)$ and $i\in [q-2]$ and 
set $\sigma$ to be the $i$th element in $\Sigma_q\setminus \{x_{n-3},x_{n-2}\}$.
Then set $\psi(\bx,i)=x_1x_2\ldots x_{n-2}\sigma x_{n-2}$.

\begin{proposition}\label{prop:recursion}
The maps $\phi$ and $\psi$ are bijections.
\end{proposition}

\begin{proof}
We construct the inverse map for $\phi$. Specifically, we set $\phi^{-1}:\Irr_{\le{2}}^{(s)}(3,n,q)\to \Irr_{\le{2}}(n-1,q)\times [q-2]$ such that $\phi^{-1}(\bx)=(x_1\ldots x_{n-1},i)$, where $i$ is the index of $x_n$ in 
$\Sigma_q\setminus\{x_{n-2},x_{n-1}\}$.
It can be verified that $\phi\circ \phi^{-1}$ and $\phi^{-1}\circ \phi$ are identity maps on their respective sets.
Similarly, the inverse map for $\psi$ is given by $\psi^{-1}:\Irr_{\le{2}}^{(s)}(2,n,q)\to \Irr_{\le{2}}(n-2,q)\times [q-2]$ such that $\psi^{-1}(\bx)=(x_1\ldots x_{n-2},i)$, where $i$ is the index of $x_{n-1}$ in 
$\Sigma_q\setminus\{x_{n-3},x_{n-2}\}$.
\end{proof}

The following corollary is then immediate.

\begin{corollary}
We have that $I_{{\le}2}(2,q)=q(q-1)$, $I_{{\le}2}(3,q)=q(q-1)^2$, and  
\begin{equation}\label{recursion:i2}
I_{{\le}2}(n,q)=
(q-2)I_{{\le}2}(n-1,q)+(q-2)I_{{\le}2}(n-2,q)
\end{equation}
for $n\ge 4$. Therefore, the asymptotic rate is $\rate_{{\le}2}(q)=\log_q \lambda_2$, 
where $\lambda_2=(q-2+ \sqrt{q^2 -4})/2$. 
\end{corollary}

In the next section, we are interested in irreducible words with certain prefixes or suffixes.
Specifically, let $\bp$ be a word of length $\ell<n$. 
Then we denote the set of irreducible words of length $n$ with prefix $\bp$ by $\Irr^{(p)}_{\lek}(\bp,n,q)$.
The set of irreducible words of length $n$ with suffix $\bp$ is denoted by $\Irr^{(s)}_{\lek}(\bp,n,q)$.

Fix $\bp$. Notice that the maps $\phi$ and $\psi$ simply appends one and two symbols, respectively, to words in their domains.
Hence, if we apply the maps to a word with prefix $\bp$, the image also has the same prefix $\bp$.
Therefore, both $\phi$ and $\psi$ remain as bijections when we restrict the domains and codomains
to the irreducible words with prefix $\bp$. 
In other words, we obtain a similar recursion for $\Irr^{(p)}_{{\le}2}(\bp,n,q)$.

\begin{corollary}\label{cor:recursion-prefix}
Let $\bp\in\Sigma_q^\ell$
For $n\ge \ell+2$,
\begin{align}
\left|\Irr^{(p)}_{{\le}2}(\bp,n,q) \right| &=
(q-2)\left|\Irr^{(p)}_{{\le}2}(\bp,n-1,q) \right| \notag\\
&~~~+(q-2) \left|\Irr^{(p)}_{{\le}2}(\bp,n-2,q) \right|. \label{recursion:p2}
\end{align}
\end{corollary}

We provide the recursion for $\Irr_{{\le}3}(n,q)$.

\begin{proposition}
We have that $I_{{\le}3}(3,q)=q(q-1)^2$, $I_{{\le}3}(4,q)=q^2(q-1)(q-2)$, $I_{{\le}3}(5,q)=q(q-1)(q-2)(q^2-q-1)$
and  
\begin{align}
I_{{\le}3}(n,q)&
= (q-2)I_{{\le}3}(n-1,q)+(q-3)I_{{\le}3}(n-2,q)\notag\\
&~~~+(q-2)I_{{\le}3}(n-3,q)\label{recursion:i3}
\end{align}
for $n\ge 6$. Therefore, $\rate_{{\le}3}(q)=\log_q \lambda_3$, 
where $\lambda_3$ is the largest real root of equation $x^3-(q-2)x^2-(q-3)x-(q-2)=0$.
\end{proposition}
\begin{proof}
Recall that for a word $\bp$ of length $\ell<n$, $\Irr^{(s)}_{\lek}(\bp,n,q)$ is the set of irreducible words of length $n$ with suffix $\bp$. 
Let $\bL\subseteq \Sigma_q^\ell$ be a set of suffixes.
We let $\Irr^{(s)}_{\lek}(\bL,n,q)$ denote the set of irreducible words of length $n$ with suffixes in $\bL$. 

To prove the proposition, we partition the set of irreducible words $\Irr_{\le{3}}(n,q)$ into three classes and provide bijections from irreducible words of shorter length into them. 
Specifically, we consider all possible suffixes of length six of a ${\le}3$-irreducible word. 
For a word $\bx\in \Irr_{\le{3}}(n,q)$, if $x_n=x_{n-3}=a$, then its suffix of length six must be of the form $\{bcabda, bcacba, bcacda, abacba, abacda\}$, where $a,b,c,d$ are distinct elements in $\Sigma_q$. 
On the other hand, if $x_n \neq x_{n-3}$, every suffix of length four must be of the form $\{abcd,abcb,abac\}$.
As such, we set
\begin{align*}
 \bL_1 &\triangleq \{abcd,abcb,abac \mid a,b,c,d \mbox{ distinct in } \Sigma_q\},\\ 
 \bL_2 &\triangleq \{bcabda,bcacba,bcacda \mid a,b,c,d \mbox{ distinct in } \Sigma_q\},\\
 \bL_3 &\triangleq\{abacba,abacda \mid a,b,c,d \mbox{ distinct in } \Sigma_q\}.
\end{align*}
We consider the following maps for $n\ge 6$.
\begin{align*}
\varphi_1: \Irr_{\le{3}}(n-1,q)\times [q-2]&\to\Irr_{\le{3}}^{(s)}(\bL_1,n,q),\\
\varphi_2: \Irr_{\le{3}}(n-2,q)\times [q-2]&\to\Irr_{\le{3}}^{(s)}(\bL_2,n,q),\\
\varphi_3: \Irr_{\le{3}}(n-3,q)\times [q-3]&\to\Irr_{\le{3}}^{(s)}(\bL_3,n,q).
\end{align*}
Recall that $\Irr_{\le{k}}(n,q)=\Irr_{\le{k}}^{(s)}(3,n,q) \cup \Irr_{\le{k}}^{(s)}(2,n,q).$
We first define $\varphi_1$. If $\bx=x_1\ldots x_{n-3}x_{n-2}x_{n-1}\in \Irr_{\le{3}}^{(s)}(3,n-1,q)$ and $i\in [q-2]$, set $\sigma$ to be the $i$th element in $\Sigma_q\setminus \{x_{n-3},x_{n-1}\}$.
Then set $$\varphi_1(\bx,i)=\bx \sigma=x_1\ldots x_{n-3}x_{n-2}x_{n-1}{\color{red}{\sigma}}.$$
If $\bx=x_1\ldots x_{n-3}x_{n-2}x_{n-1} \in \Irr_{\le{3}}^{(s)}(2,n-1,q)$ where $x_{n-3}=x_{n-1}$ and $i\in [q-2]$, set $\sigma$ to be the $i$th element in $\Sigma_q\setminus \{x_{n-2},x_{n-1}\}$. 
$$\varphi_1(\bx,i)=\bx \sigma=x_1\ldots x_{n-1}x_{n-2}x_{n-1}{\color{red}{\sigma}}.$$
Similarly, we now define $\varphi_2, \varphi_3$ as follows. \\
If $\bx=x_1\ldots x_{n-5}x_{n-4}x_{n-3}x_{n-2}\in \Irr_{\le{3}}^{(s)}(3,n-2,q)$ and $i\in [q-3]$, we set 
$$\varphi_2(\bx,i)=\bx \sigma x_{n-3} = x_1\ldots x_{n-5}x_{n-4}x_{n-3}x_{n-2}{\color{red}{\sigma x_{n-3}}},$$
where $\sigma$ is the $i$th element in $\Sigma_q\setminus \{x_{n-5},x_{n-3},x_{n-2}\}$ if $x_{n-5} \notin \{x_{n-2},x_{n-3}\}$ or the $i$th element in $\Sigma_q\setminus \{x_{n-4},x_{n-3},x_{n-2}\}$ if $x_{n-5} \in \{x_{n-2},x_{n-3}\}$. \\
If $\bx=x_1\ldots x_{n-5}x_{n-4}x_{n-3}x_{n-2} \in \Irr_{\le{3}}^{(s)}(2,n-2,q)$ where $x_{n-4}=x_{n-2}$ and $i\in [q-3]$, set $\sigma$ to be the $i$th element in $\Sigma_q\setminus \{x_{n-5},x_{n-3},x_{n-2}\}$ and 
$$\varphi_2(\bx,i)=\bx \sigma x_{n-3} = x_1\ldots x_{n-5}x_{n-4}x_{n-3}x_{n-2}{\color{red}{\sigma x_{n-3}}}.$$ 
If $\bx=x_1\ldots x_{n-5}x_{n-4}x_{n-3}\in \Irr_{\le{3}}^{(s)}(3,n-3,q)$ and $i\in [q-2]$, set $\sigma$ to be the $i$th element in $\Sigma_q\setminus \{x_{n-5},x_{n-3}\}$. Then set
$$\varphi_3(\bx,i)=\bx \sigma x_{n-5} x_{n-3}=x_1\ldots x_{n-5}x_{n-4}x_{n-3}{\color{red}{\sigma x_{n-5} x_{n-3}}}.$$ 
If $\bx=x_1\ldots x_{n-5}x_{n-4}x_{n-3} \in \Irr_{\le{3}}^{(s)}(2,n-3,q)$ where $x_{n-5}=x_{n-3}$ and $i\in [q-2]$, set $\sigma$ to be the $i$th element in $\Sigma_q\setminus \{x_{n-4},x_{n-3}\}$. Then set $$\varphi_3(\bx,i)=\bx \sigma x_{n-4} x_{n-3}=x_1\ldots x_{n-5}x_{n-4}x_{n-3}{\color{red}{\sigma x_{n-4} x_{n-3}}}.$$ 
We can prove that $\varphi_i$ is bijection for $i\in\{1,2,3\}$ by constructing the inverse map for each $\varphi_i$. We first prove $\varphi_1$ is bijection. Specifically, we set $\varphi_1^{-1}:\Irr_{\le{3}}^{(s)}(\bL_1,n,q) \to \Irr_{\le{3}}(n-1,q)\times [q-2]$ such that $\varphi_1^{-1}(\bx)=(x_1\ldots x_{n-3}x_{n-2}x_{n-1},i)$ where $i$ is the index of $x_n$ in $\Sigma_q\setminus \{x_{n-3},x_{n-1}\}$ if $x_{n-1} \neq x_{n-3}$ or $i$ is the index of $x_n$ in $\Sigma_q\setminus \{x_{n-2},x_{n-1}\}$ otherwises. It can be verified that $\varphi_1\circ \varphi_1^{-1}$ and $\varphi_1^{-1}\circ \varphi_1$ are identity maps on their respective sets. Similarly, the inverse maps for $\varphi_2,\varphi_3$ are given by 
\begin{align*}
\varphi_2^{-1}: \Irr_{\le{3}}^{(s)}(\bL_2,n,q) &\to \Irr_{\le{3}}(n-2,q)\times [q-2]\\
\varphi_3^{-1}: \Irr_{\le{3}}^{(s)}(\bL_3,n,q) &\to \Irr_{\le{3}}(n-3,q)\times [q-3].
\end{align*}
such that $\varphi_2^{-1}(\bx)=(x_1\ldots x_{n-3}x_{n-2},i)$ where $i$ is the index of $x_{n-1}$ in 
\begin{itemize}
\item $\Sigma_q\setminus \{x_{n-5},x_{n-3},x_{n-2}\}$ if $x_{n-5} \notin \{x_{n-3},x_{n-2}\}$ or $x_{n-4} = x_{n-2}$,
\item $\Sigma_q\setminus \{x_{n-4},x_{n-3},x_{n-2}\}$ if $x_{n-5} \in \{x_{n-3},x_{n-2}\}$ and $x_{n-4}\neq x_{n-2}.$
\end{itemize}
and $\varphi_3^{-1}(\bx)=(x_1\ldots x_{n-3},i)$ where $i$ is the index of $x_{n-2}$ in 
\begin{itemize}
\item $\Sigma_q\setminus \{x_{n-5},x_{n-3}\}$ if $x_{n-5} \neq x_{n-3},$
\item $\Sigma_q\setminus \{x_{n-4},x_{n-3}\}$ if $x_{n-5} = x_{n-3}.$
\end{itemize}
We can prove that $\varphi_2,\varphi_3$ are bijections as $\varphi_2\circ \varphi_2^{-1}$, $\varphi_2^{-1}\circ \varphi_2$, $\varphi_3\circ \varphi_3^{-1}$, and $\varphi_3^{-1}\circ \varphi_3$ are identity maps on their respective sets. Since $\Irr_{\le{3}}(n,q)=\Irr_{\le{3}}^{(s)}(\bL_1,n,q) \cup \Irr_{\le{3}}^{(s)}(\bL_2,n,q) \cup \Irr_{\le{3}}^{(s)}(\bL_3,n,q)$, we have the recursion \eqref{recursion:i3}.
\end{proof}

As before, the following corollary is immediate.

\begin{corollary}\label{cor:recursion-prefix2}
Let $\bp\in\Sigma_q^\ell$
For $n\ge \ell+3$,
\begin{align*}
&\left|\Irr^{(p)}_{{\le}3}(\bp,n,q) \right| =
(q-2)\left|\Irr^{(p)}_{{\le}3}(\bp,n-1,q)\right|+ \\
&(q-3)\left|\Irr^{(p)}_{{\le}3}(\bp,n-2,q) \right|+(q-2) \left|\Irr^{(p)}_{{\le}3}(\bp,n-3,q) \right|. \label{recursion:p3}
\end{align*}
\end{corollary}

We compute the values of $\rate_\lek(q)$ for $k\in\{2,3\}$ in Table~\ref{rates}.
Let $T(n,q)$ be the largest size of an $(n,{\le}3;q)$-TD code and define 
$\tau(q)\triangleq (1/n)\limsup_{n\to\infty}\log_q T(n,q)$.
From \cite{Jain, Chee:2017}, we have that that $\rate_{{\le}3}(q)\le \tau(q)\le \rate_{{\le}2}(q)$.
Therefore, Table~\ref{rates} demonstrates that ${\C}(n,{\le}3;q)$ is {\em almost} optimal for $q\ge 5$.

\begin{table}[h!]
\centering
\begin{tabular}{|c|| c | c| c| c| c|c|}
 \hline
 $q$ & 3 & 4 & 5 & 6 & 7& 8\\ \hline
 $\rate_{{\le}2}(q)$ & 0.4380 & 0.7249 & 0.8280 & 0.8788 & 0.9081 & 0.9269 \\\hline
 $\rate_{{\le}3}(q)$ & 0.3479 & 0.7054 & 0.8208 & 0.8753 & 0.9062 & 0.9258\\\hline

 \end{tabular}
  \caption{The asymptotic information rates for $\lek$-irreducible words for $k\in\{2,3\}$}
\label{rates}
\end{table}

\section{Finite State Encoder}\label{sec:fse}

For integers $\ell<m$, an {\em $(\ell,m)$-finite state encoder} is triple $(\cS,\cE,\cL)$, 
where $\cS$ is a set of {\em states},
$\cE\subset \cS\times\cS$ is a set of {\em directed edges}, and 
$\cL:\cE\to \Sigma_q^\ell\times \Sigma_q^m$ is an {\em edge labeling}.

To encode irreducible words, we choose $m\ge2k-1$, and set 
\[ \cS \triangleq \Irr_\lek(m,q) \mbox{ and }
\cE \triangleq \{(\bx,\bx') : \bx\bx'\in \Irr_\lek(2m,q)\}.
\]
For $\bx\in \cS$, we define the {\em neighbours} of $\bx$ to be 
$N(\bx)\triangleq \{\bx' : (\bx,\bx')\in \cE\}$. We also consider the quantity
$\Delta_{\lek}(m,q)\triangleq \min \{|N(\bx)|: \bx\in\cS\}$ and choose $\ell$ 
such that
\begin{equation}\label{constraint:upper}
\Delta_{\lek}(m,q)\ge q^\ell.
\end{equation}

We now define the edge labelling $\cL$ using this choice of $\ell$.
For $\bx\in\cS$, since $|N(\bx)|\ge q^\ell$, we may use the set $\Sigma^\ell$ 
to index the first $q^\ell$ words in $N(\bx)$. 
Hence, for $\bx'\in S$, if $\bx'$ is one of the first $q^\ell$ words, we let $\by_{\bx'}\in \Sigma^\ell$ denote the index.
Otherwise, we simply set $\by_{\bx'}=-$. 
Therefore, for $(\bx,\bx')\in \cE$, we set $\cL(\bx,\bx')=(\by_{\bx'},\bx')$.
Finally, we call this triple an {\em $(\ell,m)$- finite state encoder for irreducible words}.
\begin{example}\label{exa:fsm1}
Let $k=2$, $q=3$, $m=3$. 
Then $\cS=\{010, 012, 020$, $021, 101, 102, 120, 121, 201, 202, 210, 212\}$, and 
\begin{align*}
N(010)&=\{201, 210, 212\},\\
N(012) &=\{010, 012, 021, 101, 102\}.
\end{align*}
We verify that $\Delta_{{\le}2}(3,3)=3$ and so, we choose $\ell=1$.
So, we can set $\cL$ to map the edges exiting the state $010$ 
as follow:
\vspace{-5mm}
 
{\small
\[
(010,201)\mapsto (0,201),~
(010,210)\mapsto (1,210),~
(010,212)\mapsto (2,212).
\]
}
We represent the mapping $\cL$ using the following lookup table.
\begin{center}
\begin{tabular}{|c||c|c|c|c|c|}
\hline
$\bx$ & \multicolumn{5}{c|}{$N(\bx)$}\\ \hline
 & $0$ & $1$ & $2$ & -- & -- \\ \hline
010 & 201 & 210 & 212 & -- & -- \\
012 & 010 & 012 & 021 & 101 & 102 \\
020 & 102 & 120 & 121 & -- & -- \\
021 & 012 & 020 & 021 & 201 & 202 \\
101 & 201 & 202 & 210 & -- & -- \\
102 & 010 & 012 & 101 & 102 & 120 \\
120 & 102 & 120 & 121 & 210 & 212 \\
121 & 012 & 020 & 021 & -- & -- \\
201 & 020 & 021 & 201 & 202 & 210 \\
202 & 101 & 102 & 120 & -- & -- \\
210 & 120 & 121 & 201 & 210 & 212 \\
212 & 010 & 012 & 021 & -- & -- \\
\hline
\end{tabular}
\end{center}
Here, to determine $\cL(\bx,\bx')$, we look at the row corresponding to $\bx$ and 
look at the column corresponding to $\bx'$. If the column is $\by_{\bx'}$, then 
$\cL(\bx,\bx')=(\by_{\bx'},\bx')$. So, $\cL(012,010)=(0,010)$.
\end{example}

\subsection{Encoding}

Let $s$ be a positive integer and set $n=s\ell$. 
Suppose the message $\by=\by_1\by_2\ldots\by_s\in \Sigma^{s\ell}$.

To encode $\by$ using an $(\ell,m)$-finite state encoder for irreducible words, 
we do the following:
\begin{enumerate}[(I)]
\item Set $\bx_0$ to the first word in $\cS=\Irr_{\lek}(m,q)$.
\item For $i\in [s]$, set $\bx_i$ to be the unique word such that $\cL(\bx_{i-1},\bx_i)=(\by_i,\bx_i)$.
\item The encoded irreducible word is $\bx=\bx_1\bx_2\ldots \bx_s$. 
\end{enumerate}

\begin{example}[Example \ref{exa:fsm1} continued] Let $s=3$ and consider the message $\by=012$.
First, we set $\bx_0=010$. Then $\bx_1=201$ since $\cL(010,201)=(0,201)$.
Similarly, $\bx_2=021$ and $\bx_3=021$.

Therefore, the encoded word $\bx$ is $201021021$.
\end{example}

Since the encoded word has length $sm$, 
the $(\ell,m)$-finite state encoder for irreducible words has rate $\ell/m$.
In the next subsection, we see that $\ell$ and $m$ can be chosen in such a way that
the rate $\ell/m$ approaches $\rate_\lek(q)$ {\em quickly}. 

\subsection{Approaching the Asymptotic Information Rate}

Pick $\epsilon>0$. We find suitable values for $\ell$ and $m$ so that the encoding rate 
satisfies 
\begin{equation}\label{constraint:lower}
\ell/m\ge \rate_\lek(q)-\epsilon.
\end{equation}

In particular, we show that $\ell=\Theta(1/\epsilon)$ and $m=\Theta(1/\epsilon)$ suffice to guarantee \eqref{constraint:lower}. 

Recall that $\ell$ and $m$ are required to satisfy \eqref{constraint:upper}. 
Hence, we determine $\Delta_\lek(m,q)$. 
Surprisingly, these values have the same recursive structure as $I_\lek(m,q)$ and 
therefore, have the same growth rate.

\begin{proposition}
We have that $\Delta_{{\le}2}(3,q)=q(q-2)^2$, $\Delta_{{\le}2}(4,q)=(q-2)^2(q^2-q-1)$, and 
for $m\ge 5$,
\begin{equation}\label{recursion:D2}
\Delta_{{\le}2}(m,q)=
(q-2)\Delta_{{\le}2}(m-1,q)+(q-2)\Delta_{{\le}2}(m-2,q).
\end{equation}
\end{proposition}

\begin{proof}Observe that by symmetry, we have $|N(\bx)|=|N(\bx')|$ for $\bx,\bx'\in\Irr_{\le{2}}^{(s)}(2,m,q)$. 
Similarly, $|N(\by)|=|N(\by')|$ for $\by,\by'\in \Irr_{\le{2}}^{(s)}(3,m,q)$.

We first show that $|N(\bx)|\le |N(\by)|$ for $\bx\in\Irr_{\le{2}}^{(s)}(2,m,q)$ and $\by\in\Irr_{\le{2}}^{(s)}(3,m,q)$.
Without loss of generality, we assume $\bx\in\Irr_{\le{2}}^{(s)}(010,m,q)$ and $\by\in\Irr_{\le{2}}^{(s)}(210,m,q)$.
Then the neighbours of $\bx$ and $\by$ are given by
\begin{align}
N(\bx) &=\left\{\bx' : 10\bx' \in \bigcup_{\sigma \notin \{0,1\}} \Irr_{{\le}2}^{(p)}(10\sigma,m+2,q)\right\},\label{eq:Nx}\\
N(\by) &=\left\{\by' : 10\by' \in \bigcup_{\sigma \ne 0} \Irr_{{\le}2}^{(p)}(10\sigma,m+2,q)\right\}.\label{eq:Ny}
\end{align}
Since $N(\bx)\subseteq N(\by)$, the inequality $|N(\bx)|\le |N(\by)|$ follows. 
Hence, $\Delta_{{\le}2}(m,q)=|N(\bx)|$ where $\bx\in\Irr_{\le{2}}^{(s)}(010,m,q)$.

Since $\Delta_{{\le}2}(m,q)=\sum_{\sigma \notin \{0,1\}} \left|\Irr_{{\le}2}^{(p)}(10\sigma,m+2,q)\right|$,
the recursive equation \eqref{recursion:D2} follows from Corollary~\ref{cor:recursion-prefix}.
\end{proof}

For $k=3$, we have the following recursive equation.
\begin{proposition}
We have that 
\begin{align*}
\Delta_{{\le}3}(5,q)&=(q-2)(q^2-2q-1)^2,\\
\Delta_{{\le}3}(6,q)&=(q-1)(q^5-6q^4+9q^3+4q^2-8q-9),\\
\Delta_{{\le}3}(7,q)&=(q-2)(q^6-6q^4+9q^3+4q^2-8q-10q+3),
\end{align*}
\noindent and for $m\ge 8$,
\begin{align}
\Delta_{{\le}3}(m,q)&=
(q-2)\Delta_{{\le}3}(m-1,q)+(q-3)\Delta_{{\le}3}(m-2,q)\notag\\
&~~~+(q-2)\Delta_{{\le}3}(m-3,q).\label{recursion:D3}
\end{align}
\end{proposition}
\begin{proof}
Let $\bL$ be the set of all possible suffixes of length five of an irreducible word.
We can then verify that $\bL=\{abcab,abcac,abcad,abcba,abcbd,abaca,abacb,abacd,abcde,$
\\ $abcdb,abcdc,abcda \mid a,b,c,d,e \mbox{ distinct in } \Sigma_q\}$.
We first show that $\Delta_{{\le}3}(m,q)=|N(\bx)|$ where $\bx \in \Irr_{\leq 3}^{(s)}(abcab,m,q)$. 
In other words, we need to show that 
$|N(\bx)|\le |N(\by)|$ for $\bx\in\Irr_{\le{3}}^{(s)}(abcab,m,q)$ and $\by\in\Irr_{\le{3}}^{(s)}(\bp,m,q)$, where $\bp \in \bL$. 
We demonstrate the inequality in the case when $\bp=abcad$, and the remaining cases can be done similarly. 
Without loss of generality, we assume that $\bx \in \Irr_{\le{3}}^{(s)}(01201,m,q)$ and $\by \in \Irr_{\le{3}}^{(s)}(03201,m,q)$. Then the neighbours of $\bx$ and $\by$ are given by
\begin{align}
N(\bx) &=\left\{\bx' : 201\bx' \in \bigcup_{\sigma \notin \{1,2\}} \Irr_{{\le}3}^{(p)}(201\sigma,m+3,q)\right\},\label{eq:Nx}\\
N(\by) &=\left\{\by' : 201\by' \in \bigcup_{\sigma \ne 1} \Irr_{{\le}3}^{(p)}(201\sigma,m+3,q)\right\}.\label{eq:Ny}
\end{align}
Since $N(\bx)\subseteq N(\by)$, the inequality $|N(\bx)|\le |N(\by)|$ follows. Hence, $\Delta_{{\le}3}(m,q)=|N(\bx)|$ where $\bx\in\Irr_{\le{3}}^{(s)}(01201,m,q)$.

Since $\Delta_{{\le}3}(m,q)=\sum_{\sigma \notin \{1,2\}} \left|\Irr_{{\le}3}^{(p)}(201\sigma,m+3,q)\right|$,
the recursive equation \eqref{recursion:D3} follows from Corollary~\ref{cor:recursion-prefix2}.
\end{proof}

Recall that $\lambda_2$ and $\lambda_3$ are roots of the equations
$x^2-(q-2)x-(q-2)=0$ and $x^3-(q-2)x^2-(q-3)x-(q-2)=0$, respectively.

Set $\kappa_2$ such that $\Delta_{{\le}2}(m,q)\ge \kappa_2\lambda_2^m$ for $m\in\{3,4\}$.
Similarly, set $\kappa_3$ so that  $\Delta_{{\le}3}(m,q)\ge \kappa_3\lambda_3^m$ for $m\in\{5,6,7\}$.
Then it follows from an inductive argument and recursions \eqref{recursion:D2} and \eqref{recursion:D3} that 
\begin{equation}\label{Delta}
\Delta_{\lek}(m,q)\ge \kappa_k\lambda_k^m \mbox{ for all }m.
\end{equation}
We are now ready to present the main theorem of this section.

\begin{theorem}Let $k\in\{2,3\}$. Set $c_k=\rate_{\lek}(q)=\log_q\lambda_k$.
For $\epsilon>0$, if we choose $m$ and $\ell$ such that
\begin{align}
\ell&=\ceil*{\frac{(c_k-\epsilon)(c_k-\log_q{\kappa_k})}{\epsilon}},\label{value:ell}\\
m&=\ceil*{\frac{\ell-\log_q{\kappa_k}}{c_k}},\label{value:m}
\end{align}
then the $(\ell,m)$-finite state encoder has rate at least $\rate_{\lek}(q)-\epsilon$.
\end{theorem}

\begin{proof}We have to verify that \eqref{constraint:upper} and \eqref{constraint:lower} 
hold for the choice of $\ell$ and $m$.
Now, \eqref{value:ell} implies that $\epsilon\ell\ge (c_k-\epsilon)(c_k-\log_q{\kappa_k})$, and equivalently, 
$c_k\ell/(\ell-\log_q\kappa_k+c_k)\ge c_k-\epsilon$.
Therefore, 
\[\frac{\ell}{m} \ge \frac{\ell}{1+(\ell-\log_q\kappa_k)/c_k}=\frac{c_k\ell}{\ell-\log_q\kappa_k+c_k}\ge c_k-\epsilon.
\]
Thus, we verify \eqref{constraint:lower}.

Next, from \eqref{Delta} and \eqref{value:m}, we have that 
\[ \Delta_{\lek}(m,q)\ge \kappa_k\lambda_k^{(\ell-\log_q\kappa_k)/ \log_q\lambda_k}=q^\ell.\]
Hence, we verify \eqref{constraint:upper} and complete the proof.
\end{proof}

Therefore, to achieve encoding rates at least $\rate_{\lek}(q)-\epsilon$, we only require $\ell=\Theta(1/\epsilon)$ and $m=\Theta(1/\epsilon)$. 
If we naively use a lookup table to represent $(\cS,\cE,\cL)$, we require $q^{\Theta(1/\epsilon)}$ space.
Furthermore, using binary search, the $(\ell,m)$-finite state encoder for irreducible words encodes in $O(n/\epsilon)$ time.
In the next section, we use combinatorial insights from \eqref{recursion:i2} and \eqref{recursion:i3} 
to reduce the space requirement to $O(1/\epsilon^2)$. 

\section{Ranking/Unranking algorithm}\label{sec:rank}

A {\em ranking function} for a finite set $S$ of cardinality $N$ is a bijection 
${\rm rank}:S\rightarrow [N]$.
Associated with the function {\rm rank} is a unique {\em unranking function}
${\rm unrank}:[N]\rightarrow S$,
such that ${\rm rank}(s)=j$ if and only if ${\rm unrank}(j)=s$ for all $s\in S$ and
$j\in[N]$. In this section, we present an algorithm for
ranking and unranking $\Irr_{\lek}(n,q)$. 
For ease of exposition, we focus on the case where $k=2$ and present the ranking/unranking algorithm for $k=3$ at the end of the section.
The basis of our ranking and unranking algorithms is the bijections defined in Section~\ref{sec:enumerate}. As implied by the codomains of $\phi$ and $\psi$, for $n\ge 4$, we order the words in $\Irr_{{\le}2}(n,q)$ 
such that words in $\Irr_{{\le}2}^{(s)}(3,n,q)$ are ordered before words in $\Irr_{{\le}2}^{(s)}(2,n,q)$.
For words in $\Irr_{{\le}2}(2,q)$ and $\Irr_{{\le}2}(3,q)$, we simply order them lexicographically.  
We illustrate the idea behind the unranking algorithm through an example.

\begin{example}
Let $n=6$ and $q=3$. Then the values of $I_{{\le}2}(m,q)$ are as follow.
\begin{center}
\begin{tabular}{|c||c|c|c|c|c|}
\hline
$m$ & 2 & 3 & 4 & 5 & 6\\ \hline
$I_{{\le}2}(m,q)$ & 6 & 12 & 18 & 30 & 48\\ \hline
\end{tabular}
\end{center}
Suppose we want to compute ${\rm unrank}(40)$. Proposition~\ref{prop:recursion} gives
\begin{equation*}
\Irr_{{\le}2}(6,3) = \phi(\Irr_{{\le}2}(5,3)\times [1]) \cup \psi(\Irr_{{\le}2}(4,3)\times[1]). 
\end{equation*}
Now, we are interested in the 40th word of $\Irr_{{\le}2}(6,3)$.
Since $40>I_{{\le}2}(5,3)=30$, the 40th word of $\Irr_{{\le}2}(6,3)$ is 
the image of the $40-30=10$-th word in $\Irr_{{\le}2}(4,3)$ under $\psi$. 
Recursing tells us that the $10$-th word in $\Irr_{{\le}2}(4,3)$ is the $10$-th element in $\phi(\Irr_{{\le}2}(3,3)\times [1])$.
The $10$-th element of $\Irr_{{\le}2}(3,3)$ is $202$. This gives
\begin{align*}
{\rm unrank}(40) &= \psi(\phi(202,1),1) \\
&= \psi(202{\color{red}{1}},1) 
= 2021{\color{red}{01}}.
\end{align*}
\end{example}
The formal unranking algorithm is described in Algorithm~\ref{alg1}.

\begin{algorithm}
\small
\caption{${\tt unrank}(n,q,j)$}\label{alg1}
\begin{algorithmic}
\REQUIRE Integers $n \geq 2$, $q\ge 3$, $1\leq j\leq I_{{\le}2}(n,q)$
\ENSURE $\bx$, where $\bx$ is the codeword of rank $j$ in $\Irr_{{\le}2}(n,q)$
\vspace{0.05in}
\IF{$n\leq 3$} \RETURN{$j$-th codeword in $\Irr_{{\le}2}(n,q)$} \ENDIF

\IF{$j\le(q-2)I_{{\le}2}(n-1,q)$} 
\STATE{$j'\gets 1+\floor*{(j-1)/(q-2)}$}
\STATE{$i\gets (j-1)\pmod{q-2}+1$}
\RETURN{$\phi({\tt unrank}(n-1,q,j'),i)$}
\ELSE
\STATE{$j'\gets 1+\floor*{(j-(q-2)I_{{\le}2}(n-1,q)-1)/(q-2)}$}
\STATE{$i\gets (j-(q-2)I_{{\le}2}(n-1,q)-1)\pmod{q-2}+1$}
\RETURN{$\psi({\tt unrank}(n-2,q,j'),i)$} \ENDIF
\end{algorithmic}
\end{algorithm}

The corresponding ranking algorithm for $\Irr_{{\le}2}(n,q)$ has a similar recursive structure and is described
in Algorithm~\ref{alg2}.

\begin{example}
Let $n=6$ and $q=3$ as before.
Suppose we want to compute ${\tt rank}(202101)$. 
Since $202101\in\Irr^{(s)}_{{\le}2}(2,6,3)$, 
we have that $202101$ is obtained from applying $\psi$ to $2021\in \Irr_{{\le}2}(4,3)$.
Again, since $2021\in\Irr^{(s)}_{{\le}2}(3,6,3)$, 
we have that $202$ is obtained from applying $\phi$ to $202\in \Irr_{{\le}2}(3,3)$.
Therefore, 
\begin{align*}
{\rm rank}(202101)&= {\rm rank}(2021)+I_{{\le}2}(5,3)\\
&= {\rm rank}(202)+I_{{\le}2}(5,3)\\
&=10+30=40
\end{align*}
\end{example}
\begin{algorithm}
\small
\caption{${\tt rank}(n,q,\bx)$}\label{alg2}
\begin{algorithmic}
\REQUIRE $n \geq 2$, $q\ge 3$ and irreducible word $\bx$ of length $n$  
\ENSURE $j$, where $1 \leq j \leq I_{{\le}2}(n,q)$, the rank of $\bx$ in $\Irr_{{\le}2}(n,q)$
\vspace{0.05in}
\IF{$n\leq 3$} \RETURN{${\tt rank}(\bx)$ in ${\Irr}_{{\le}2}(n,q)$} \ENDIF
\IF{$x_{n}\ne x_{n-2}$}
\STATE{$\bx'\gets x_1x_2\ldots x_{n-1}$}
\STATE{$i\gets$ the index of $x_{n}$ in $\Sigma_q\setminus\{x_{n-2},x_{n-1}\}$}
\RETURN{$({\tt rank}(n-1,q,\bx')-1)(q-2)+i$} 
\ELSE
\STATE{$\bx'\gets x_1x_2\ldots x_{n-2}$}
\STATE{$i\gets$ the index of $x_{n-1}$ in $\Sigma_q\setminus\{x_{n-3},x_{n-2}\}$}
\RETURN{$({\tt rank}(n-2,q,\bx')-1)(q-2)+i+(q-2)I_{{\le}2}(n-1,q)$} 
\ENDIF
\end{algorithmic}
\end{algorithm}
\noindent The set of values of $\{I_{{\le}2}(m,q): m\le n\}$ required in Algorithms~\ref{alg1} and~\ref{alg2} 
can be precomputed based on the recurrence \eqref{recursion:i2}.
Since the numbers $I_{{\le}2}(m,q)$ grow exponentially, these $n$ stored values require $O(n^2)$ space.

Next, Algorithms~\ref{alg1} and~\ref{alg2} involve $O(n)$ iterations and
each iteration involves a constant number of arithmetic operations.
Therefore, Algorithms~\ref{alg1} and~\ref{alg2} involve $O(n)$ arithmetics operations and 
have time complexity $O(n^2)$. Similarly, the corresponding ranking/unranking algorithm for $\Irr_{{\le}3}(n,q)$ have similar recursive structures and are described
in Algorithm~\ref{alg3} and~\ref{alg4}.

\begin{algorithm}
\small
\caption{${\tt unrank}(n,q,j)$}\label{alg3}
\begin{algorithmic}
\REQUIRE Integers $n \geq 3$, $q\ge 3$, $1\leq j\leq I_{{\le}3}(n,q)$
\ENSURE $\bx$, where $\bx$ is the codeword of rank $j$ in $\Irr_{{\le}3}(n,q)$
\vspace{0.05in}
\IF{$n\leq 5$} \RETURN{$j$-th codeword in $\Irr_{{\le}3}(n,q)$} \ENDIF

\IF{$j\le(q-2)I_{{\le}3}(n-1,q)$} 
\STATE{$j'\gets 1+\floor*{(j-1)/(q-2)}$}
\STATE{$i\gets (j-1)\pmod{q-2}+1$}
\RETURN{$\varphi_1({\tt unrank}(n-1,q,j'),i)$}
\ELSE
\STATE{$j'\gets j-(q-2)I_{{\le}3}(n-1,q)$}
\ENDIF
\IF{$j'\le(q-3)I_{{\le}3}(n-2,q)$} 
\STATE{$j'\gets 1+\floor*{(j'-1)/(q-3)}$}
\STATE{$i\gets (j'-1)\pmod{q-3}+1$}
\RETURN{$\varphi_2({\tt unrank}(n-2,q,j'),i)$}
\ELSE
\STATE{$j'\gets j'-(q-3)I_{{\le}3}(n-2,q)$}
\STATE{$i\gets (j'-1)\pmod{q-2}+1$}
\RETURN{$\varphi_3({\tt unrank}(n-3,q,j'),i)$} \ENDIF
\end{algorithmic}
\end{algorithm}

\begin{algorithm}
\small
\caption{${\tt rank}(n,q,\bx)$}\label{alg4}
\begin{algorithmic}
\REQUIRE $n \geq 3$, $q\ge 3$ and irreducible word $\bx$ of length $n$  
\ENSURE $j$, where $1 \leq j \leq I_{{\le}3}(n,q)$, the rank of $\bx$ in $\Irr_{{\le}3}(n,q)$
\vspace{0.05in}
\IF{$n\leq 5$} \RETURN{${\tt rank}(\bx)$ in ${\Irr}_{{\le}3}(n,q)$} \ENDIF
\IF{$\bx \in \Irr_{\le{3}}^{(s)}(\bL_1,n,q)$}
\STATE{$(\bx',i) \gets \varphi_1^{-1}(\bx)$}
\RETURN{$({\tt rank}(n-1,q,\bx')-1)(q-2)+i$} \ENDIF

\IF{$\bx \in \Irr_{\le{3}}^{(s)}(\bL_2,n,q)$}
\STATE{$(\bx',i) \gets \varphi_2^{-1}(\bx)$}
\RETURN{$({\tt rank}(n-2,q,\bx')-1)(q-3)+i+(q-2)I_{{\le}3}(n-1,q)$} \ENDIF

\IF{$\bx \in \Irr_{\le{3}}^{(s)}(\bL_3,n,q)$}
\STATE{$(\bx',i) \gets \varphi_3^{-1}(\bx)$}
\RETURN{$({\tt rank}(n-3,q,\bx')-1)(q-2)+i+(q-2)I_{{\le}3}(n-1,q)+(q-3)I_{{\le}3}(n-2,q)$} \ENDIF
\end{algorithmic}
\end{algorithm}

\subsection{Reducing the Space Requirement for the Finite State Encoder}

As discussed earlier, a naive implementation of the $(\ell,m)$-finite state encoder in Section~\ref{sec:fse}
requires $q^{\Theta(m)}$ space (assuming $\ell=\Theta(m)$).
Here, we modify our unranking algorithm to reduce the space requirement $O(m)$ integers or $O(m^2)$ bits.

Recall the notation in Section~\ref{sec:fse}. 
We discuss only  for the case $k=2$ as the case $k=3$ is similar.
In particular, let $\bx_{i-1}\in \Irr_{{\le}2}(m,q)$ and $\by_{i}\in \Sigma_q^\ell$. 
Our encoding task is to determine the irreducible word $\bx_{i}$ in $N(\bx_i)$ whose index corresponds to $\by_i$.
Equivalently, if $j$ is the rank of $\by_{i}\in \Sigma_q^\ell$, then our task is to find $\bx_i$ such that its rank in $N(\bx_{i-1})$ is $j$. Since $\bx_{i-1}$ is irreducible and using symmetry, we assume that 
$\bx_{i-1}\in\Irr_{\le{2}}^{(s)}(010,m,q)$ or $\bx_{i-1}\in\Irr_{\le{2}}^{(s)}(210,m,q)$.
Furthermore, \eqref{eq:Nx} and \eqref{eq:Ny} imply that $N(\bx_{i_1})$ corresponds to
a union of ${\le}2$-irreducible words with prefixes of the form $10\sigma$.
Therefore, it suffices to provide ranking/unranking algorithms for $\Irr_{\le{2}}^{(p)}(10\sigma,m,q)$.

Since \eqref{recursion:p2} implies that $\Irr_{\le{2}}^{(p)}(10\sigma,m,q)$ has the same recursive structure as 
$\Irr_{\le{2}}(m,q)$, we can modify Algorithms~\ref{alg1} and~\ref{alg2} to unrank and rank $\Irr_{\le{2}}^{(p)}(10\sigma,m,q)$.

Now, to rank/unrank $\Irr_{\le{2}}^{(p)}(10\sigma,m,q)$ require $O(m)$ precomputed integers.
Assuming $q$ is constant, we require only $O(m)$ integers or $O(m^2)$ bits.
However, the running time is increased to $O(m^2)$.

\section{Conclusion}
For $k\in\{2,3\}$ and all $q$, we provided an explicit recursive formula for $\Irr_{\lek}(n,q)$ and 
hence, derived the expressions for $\rate_{\lek}(q)$.

We design efficient encoders/decoders for $\Irr_{\lek}(n,q)$.
\begin{enumerate}[(i)]
\item We provide an $(\ell,m)$-finite state encoder
and showe that for all $\epsilon>0$, if we choose $m=\Theta(1/\epsilon)$ and $\ell=\Theta(1/\epsilon)$, 
the encoder achieves rate that is at least $\rate_{\lek}(q)-\epsilon$.
The implementation of the finite state encoder with a lookup table runs in $O(n/\epsilon)$ time and requires $q^{\Theta(1/\epsilon)}$ space. 
However, if we use the ranking/unranking method in Section~\ref{sec:rank}, the encoder
runs in $O(n/\epsilon^2)$ time and requires $O(1/\epsilon)$ space. 
\item We provide an unranking algorithm for irreducible words whose encoding rate is $(1/n)\log_q({\rm Irr}_{\lek}(n,q))\ge\rate_{\lek}(q)$. The encoder runs in $O(n^2)$ time and requires $O(n^2)$ space.
\end{enumerate}

\end{document}